\documentclass[preprint]{elsarticle}

\usepackage{array,xspace,multirow,hhline,tikz,colortbl,tabularx,booktabs,fixltx2e,amsmath,amssymb,amsfonts,amsthm}
\usepackage{algorithm}
\usepackage{algorithmic}
\usepackage{verbatim,ifthen}
\usepackage{enumitem}
\usepackage{pifont}
\usepackage{ifthen}
\usepackage{calrsfs,mathrsfs}
\usepackage{bbding,pifont}
\usepackage{pgflibraryshapes}

\usetikzlibrary{trees}
\usetikzlibrary{positioning,chains,fit,shapes,calc}
\usepackage{tkz-graph}

\usepackage{eucal}

\usepackage{tikz}
\usetikzlibrary{arrows}
\usetikzlibrary{decorations.pathreplacing}

	\usepackage{varioref}

\definecolor{light-gray}{gray}{0.9}

	\newtheorem{lemma}{Lemma}%
	\newtheorem{theorem}{Theorem}%
	\newtheorem{corollary}{Corollary}%

			\newtheorem{fact}{Fact}

		\newcommand{\price}{p\xspace}

	\newcommand\eat[1]{}

	\usepackage{enumitem}
	\setenumerate[1]{label=\rm(\it{\roman{*}}\rm),ref=({\it\roman{*}}),leftmargin=*}
	\newlength{\wordlength}

	\newcommand{\midd}{\mathbin{:}}

	\newcommand{\eqclass}[2][]{\ifthenelse{\equal{#1}{}}{[#2]}{[#2]_{\sim_{#1}}}}

	\newcommand{\Pref}[1][]{
		\ifthenelse{\equal{#1}{}}{\mathrel R}{\mathop{R_{#1}}}
	}                                          
	\newcommand{\sPref}[1][]{                  
		\ifthenelse{\equal{#1}{}}{\mathrel P}{\mathop{P_{#1}}}
	}                                          
	\newcommand{\Indiff}[1][]{                 
		\ifthenelse{\equal{#1}{}}{\mathrel I}{\mathop{I_{#1}}}
	}
	\newcommand{\prefset}[1][]{\ifthenelse{\equal{#1}{}}{\mathcal{R}}{\mathcal{R}_{#1}}}

\usepackage{enumitem}
\setenumerate[1]{label=\rm(\it{\roman{*}}\rm),ref=({\it\roman{*}}),leftmargin=*}

\newcommand{\nbh}[1][]{
	\ifthenelse{\equal{#1}{}}{\nu}{\nu(#1)}
}

\newcommand{\cstr}[1][]{
	\ifthenelse{\equal{#1}{}}{\mathscr S}{\cstr(#1)}
}

\newcommand{\choice}[1][]{
	\ifthenelse{\equal{#1}{}}{\mathit{C}}{\choice(#1)}
}

\usepackage{boxedminipage}
\usepackage{xspace}

\newcommand{\pbDef}[3]{%
\noindent
\begin{center}
\begin{boxedminipage}{0.98 \columnwidth}
#1\\[5pt]
\begin{tabular}{l p{0.75 \columnwidth}}
Input: & #2\\
Question: & #3
\end{tabular}
\end{boxedminipage}
\end{center}
}

\tikzset{
  treenode/.style = {align=center, inner sep=0pt, text centered,
    font=\sffamily},
  arn_n/.style = {treenode, circle, white, font=\sffamily\bfseries, draw=black,
    fill=black, text width=1.5em},
  arn_r/.style = {treenode, circle, red, draw=red, 
    text width=1.5em, very thick},
  arn_x/.style = {treenode, rectangle, draw=black,
    minimum width=0.5em, minimum height=0.5em}
}

\begin{document}

\title{Competitive Equilibrium with Equal Incomes\\ for Allocation of Indivisible Objects}
\author{Haris Aziz}
	\address{NICTA and UNSW, Sydney, Australia}



%

\begin{abstract}
We settle the complexity of computing a discrete CEEI (Competitive Equilibrium with Equal Incomes)  assignment by showing it is strongly NP-hard. We then highlight a fairness notion (CEEI-FRAC) that is even stronger than CEEI for discrete assignments, is always Pareto optimal, and can be verified in polynomial time. We also show that computing a CEEI-FRAC discrete assignment is strongly NP-hard in general but polynomial-time computable if the utilities are zero or one.

\end{abstract}


\maketitle

\section{Introduction}
CEEI (Competitive Equilibrium with Equal Incomes) is one of the most fundamental solution concepts in resource allocation~\citep{Budi11a,Moul03a}. The concept is based on the idea of a market based equilibrium: each agent has cardinal utilities over objects and each agent is considered to have equal budget of unit one to spend. Based on the utilities of the agents, an assignment of objects is in CEEI, if for some price vector, the supply meets demand and the agents use their budget to maximize utility. CEEI is an attractive solution concept because it implies envy-freeness.

In AAMAS 2014, Bouveret and Lema{\^\i}tre~\citep{BoLe14a,BoLe15a} highlighted a scale of fairness criteria. Among the criteria, they positioned CEEI~\citep{Budi11a,OSB10a} as the most stringent which implies the other fairness criteria. 
Recently, \citet{APR14a} showed that finding a discrete A-CEEI assignment is PPAD-complete and it is NP-hard to distinguish between an instance where an exact CEEI assignment exists, and one in which there is no A-CEEI tighter than guaranteed in \citet{Budi11a}.

\citet{BoLe14a,BoLe15a} posed the following open problem: what is the computational complexity of checking whether these exists a discrete CEEI assignment? We settle this problem by showing it is strongly NP-hard.
In addition to the CEEI notion previously studied~\citep{BoLe14a,BoLe15a,Budi11a,OSB10a} which we will from now on refer to as CEEI-DISC, we highlight a notion called CEEI-FRAC that is stronger than CEEI-DISC. The desirable aspect of CEEI-FRAC is that it implies the CEEI-DISC  and unlike CEEI-DISC, it implies Pareto optimality. Whereas the complexity of checking whether an assignment is CEEI-DISC has been open, we show that CEEI-FRAC can be verified in polynomial time. We also show that computing a CEEI-FRAC discrete assignment is strongly NP-hard in general but polynomial-time if the utilities are zero or one.
Independently from our work~\citep{Aziz15a}, \citet{SHM15a} recently presented complementing complexity and characterization results for CEEI-DISC and showed that checking whether a given assignment is CEEI-DISC is coNP-complete. Hence, there is a marked contrast between the complexity of testing CEEI-DISC and CEEI-FRAC.

\section{Preliminaries}
%

An assignment problem is a triple $(N,O,u)$ such that $N=\{1,\ldots, n\}$ is the set of agents, $O=\{o_1,\ldots, o_m\}$ is the set of objects, and $u=(u_1,\ldots, u_n)$ is the utility profile which specifies for each agent $i\in N$ utility function $u_{ij}$ where $u_{ij}$ denotes the utility of agent $i$ for object $o_j$. We assume that for each $j\in \{1,\ldots, m\}$,  $u_{ij}>0$ for some $i\in N$ and for each $i\in N$,  $u_{ij}>0$ for some $j\in \{1,\ldots,m\}$.


A fractional assignment $x$ is a $(n\times m)$ matrix $[x_{ij}]$ such that $ x_{ij} \in [0,1]$ for all $i\in N$, and $o_j\in O$, and  $\sum_{i\in N}x_{ij}= 1$ for all $o_j\in O$.
The value $x_{ij}$ represents the fraction of object $o_j$ being allocated to  agent $i$. Each row $x_i=(x_{i1},\ldots, x_{im})$ represents the \emph{allocation} of agent $i$. 
The set of columns correspond to the objects $o_1,\ldots, o_m$.
A feasible fractional assignment is \emph{discrete} if $x_{ij}\in \{0,1\}$ for all $i\in N$ and $o_j\in O$. 
Let the set of all fractional allocations be $\mathcal{F}$ and the set of all discrete allocations be $\mathcal{D}$.
Each agent has additive utility so that the utility received by agent $i$ from assignment $x$ is $u_i(x_i)=\sum_{o_j\in O}x_{ij}u_{ij}.$ The Nash welfare of an assignment $x$ is $\prod_{i\in N}u_i(x_i)$.


An assignment $x$ satisfies CEEI-DISC if there exists a price vector $\price=(\price_1,\ldots,\price_m)$ that specifies the price $\price_j$ of each object $o_j$ such that the maximal share that each $i\in N$ can get with budget $1$ is $x_i\in \{x'\in \mathcal{D}\midd x_i'\in \text{argmax}\{u_i(x_i'):\sum_{o_j\in O}x_{ij}'\cdot(\price_j)\leq 1\}\}$. Note that the assignment $x$ itself could be fractional. 
When we only consider discrete assignments, CEEI-DISC coincides with the CEEI notion for discrete allocations studied in \citep{BoLe14a,BoLe15a,Budi11a,OSB10a}.
An assignment $x$ satisfies CEEI-FRAC if there exists a price vector $\price=(\price_1,\ldots,\price_m)$ that specifies the price $\price_j$ of object $o_j$ such that the maximal share that each $i\in N$ can get with
 budget $1$ is $x_i\in \{x'\in \mathcal{F}\midd x_i'\in \text{argmax}\{u_i(x_i'):\sum_{o_j\in O}x_{ij}'\cdot(\price_j)\leq 1\}\}$.
CEEI-FRAC coincides with the market equilibrium notion studied in \citep{Vazi07a}.

\section{CEEI-DISC}

%

We settle the complexity of CEEI-DISC and answer the open problem in \citep{BoLe14a,BoLe15a}.

\begin{theorem}\label{th:ceei-disc}
	Computing a discrete CEEI-DISC assignment is weakly NP-hard even for two agents. 
	\end{theorem}
	\begin{proof}
		We reduce from the following weakly NP-complete problem~\citep{GaJo79a}:
	 \pbDef{\textsc{Partition}}
			 {A multiset $S$ of $m$ positive integers with total sum $2W$. }
			 {Does there exists a partition of the elements in $S$ such that the sum of integers in each set in the partition is $W$?}
			 
			 We construct an assignment instance with two agents and $m$ objects corresponding to the $m$ integers. Agents have identical utilities for the $m$ objects and the utility of each object is equal to the integer weight of the corresponding integer. 
Note that there exists a discrete assignment in which each agent gets the same utility ($W$) iff there is an partition of the integers corresponding to the utility weights so that each set of integers has total weight $W$.	Also observe that when agents have identical utility functions, a (complete) discrete assignment $x$ satisfies CEEI-DISC if and only if the utilities of all the agents are identical (Proposition 12, \citep{BoLe14a}). Hence we get that \textsc{Partition} has a yes instance iff there exists a CEEI-DISC discrete assignment. 
Since \textsc{Partition} is weakly NP-complete, it follows that computing a discrete CEEI-DISC assignment is weakly NP-hard.
		\end{proof}

Next we show that computing a discrete CEEI-DISC assignment is strongly NP-hard.

\begin{theorem}\label{th:ceei-disc}
	Computing a discrete CEEI-DISC assignment is strongly NP-hard.
	\end{theorem}
	\begin{proof}
		We reduce from the following strongly NP-complete problem~\citep{GaJo79a}:
	 \pbDef{\textsc{3-Partition}}
			 {A finite set $E=\{e_1,\ldots, e_{3n}\}$ of $3n$ elements, a bound $W$ and integer weight $w(e_j)$ for each $e_j\in E$ such that 
$\frac{W}{4}<w(e_j)<\frac{W}{2}$ and $w(E)=\sum_{j=1}^{3n}w(e_j)=nW$.}
			 {Can $E$ can be partitioned into $n$ disjoint sets $E_1,\ldots, E_m$ and weight $w(E_i)=W$ for all $i\in [n]$?}
			 
			 We construct an assignment instance with $n$ agents and $3n$ objects corresponding to the $3n$ integers. Agents have identical utilities for the $3n$ objects and the utility of each object is equal to the integer weight of the corresponding integer. 
Note that there exists a discrete assignment in which each agent gets the same utility ($W$) iff there is an integer tri-partition of the integers corresponding to the utility weights so that each set of integers has total weight $W$.	Also observe that when agents have identical utility functions, a (complete) discrete assignment $x$ satisfies CEEI-DISC if and only if the utilities of all the agents are identical (Proposition 12, \citep{BoLe14a}). Hence we get that \textsc{3-Partition} has a yes instance iff there exists a CEEI-DISC discrete assignment. 
Since \textsc{3-Partition} is strongly NP-complete, it follows that computing a discrete CEEI-DISC assignment is strongly NP-hard.
		\end{proof}

\section{CEEI-FRAC}

\citet{Vazi07a} points out that a market equilibrium allocation always exists for fractional assignments and it can be computed via the solution to the following linear program with a convex objective due to Eisenberg and Gale where each agent has budget $e_i$.
\begin{align}
\text{max} \quad 
\sum_{i=1}^{n} e_i\log {u_i} \notag &
\quad \text{s.t.} \quad
u_i = \sum_{j=1}^{m} u_{ij}x_{ij}\quad \forall i\in N \notag &\\
\sum_{i=1}^{n} x_{ij} \leq 1 &\quad \forall j\in \{1,\ldots, m\}\notag \quad 
\text{ and }x_{ij} \geq 0\quad \forall i, j. \notag \notag
\end{align}

Using duality theory, one can interpret the dual variable $p_i$ associated with the constraints $\sum_{i=1}^{n} x_{ij} \leq 1$ as the price of consuming a unit of object $o_j$~\citep{ShSo82a,Vazi07a}. 
Since in CEEI-FRAC, each agent has equal budget, we have the following facts~\citep{Vazi07a}:
%

\begin{fact}\label{fact:ceei-maxNash}
	A fractional assignment maximizing Nash welfare gives a CEEI-FRAC solution.
	\end{fact}
\begin{fact}\label{fact:same-utility}
	All CEEI-FRAC assignments give the same utility to the agents and prices to the objects.
	\end{fact}
	%
		
	
	From Facts~\ref{fact:ceei-maxNash} and ~\ref{fact:same-utility}, we get that an assignment is a CEEI-FRAC assignment if and only if maximizes Nash welfare among all fractional assignments. 
	Note that even for discrete settings, CEEI-FRAC constitutes a well-defined and attractive fairness concept. 
A CEEI-FRAC discrete assignment is an assignment that yields the same utilities as the unique utilities and prices from the fractional CEEI assignments. 

 \citet{BoLe14a} wondered whether Pareto optimality and envy-freeness is also a sufficient condition for CEEI-DISC. The next argument shows that this is not the case for discrete CEEI-FRAC assignments.

	\begin{theorem}\label{th:ef-po}
		An envy-free and Pareto optimal assignment may not be CEEI-FRAC.
		\end{theorem}
		\begin{proof}
				An assignment that is EF and PO may not be CEEI-FRAC.
				Consider the following utilities of agents:
					\begin{table}[h!]
				\centering	
				\begin{tabular}{l|llll}  
				&$o_1$&$o_2$&$o_3$&$o_4$\\\midrule
				1&$95$&$5$&$2$&$1$\\
				2&$1$&$2$&$5$&$95$
				\end{tabular}
				\vspace{-1em}
				\label{table:fractionalassigmentRSD4}
				\end{table}		
				Then both complete discrete assignments in which agent $1$ gets $o_1$ and agent $2$ gets $o_4$ are envy-free and Pareto optimal. However the assignment in which $1$ gets only $o_1$ and $2$ gets the other objects is not a CEEI-FRAC assignment. 
				\end{proof}

				Next, we show that verifying a discrete CEEI-FRAC assignment is easy whereas checking whether there exists a discrete assignment which is CEEI-FRAC is strongly NP-complete.

				\begin{theorem}\label{th:check-ceei}
					It can be checked in polynomial-time whether a given discrete assignment $y$ is a CEEI-FRAC assignment.
					\end{theorem}
					\begin{proof}
						Compute in polynomial time the utilities achieved by the agents in a CEEI-FRAC fractional solution $x$ via the Eisenberg-Gale convex program. The discrete assignment $y$ is a CEEI-FRAC assignment iff
the utilities achieved by the agents in $y$ are the same as in $x$. 
						\end{proof}

						\begin{theorem}\label{th:exists-ceei}
							Checking whether there exists a discrete assignment which is CEEI-FRAC is strongly NP-complete.
							\end{theorem}
			\begin{proof}
				Checking whether exists a discrete assignment is CEEI-FRAC is equivalent to checking whether there exists one which achieves the maximum Nash welfare that can be achieved by a fractional assignment. The problem is strongly NP-hard via the same reduction as in proof of Theorem~\ref{th:ceei-disc}.				
				\end{proof}	
				
				
								The argument above also shows that for a restricted domain, if Nash welfare maximization is in P for discrete assignments,  then checking whether there exists a CEEI-FRAC discrete assignment is also in P.

				\begin{lemma}
					If a Nash welfare maximizing discrete assignment can be computed in polynomial time, it can be checked in polynomial time whether a CEEI-FRAC discrete assignment exists or not.
					\end{lemma}
					\begin{proof}
						In order to check whether a discrete CEEI-FRAC assignment exists or not, we can compute a Nash welfare maximizing discrete assignment and check whether the Nash welfare achieved is equal to the objective acheived in the Gale-Eisenberg program.
						\end{proof}
						
						\begin{corollary}
							If agents have utilities zero or one, then it can be checked in polynomial time whether a CEEI-FRAC discrete assignment exists or not. 
							\end{corollary}
		\begin{proof}
			The statement follows from the fact that for 1-0 utilities, a Nash welfare maximizing discrete assignment can be computed in polynomial time~\citep{DaSc14a}.
			\end{proof}

				\subsubsection*{Acknowledgments}
				NICTA is funded by the Australian Government through the Department of Communications and the Australian Research Council through the ICT Centre of Excellence Program. The author thanks Sylvain Bouveret,  Michel Lema{\^\i}tre and Toby Walsh for feedback.


\end{document}